\documentclass[a4paper,twocolumn,superscriptaddress,11pt,accepted=2018-11-20]{quantumarticle}
\pdfoutput=1
\usepackage[utf8]{inputenc}
\usepackage[english]{babel}
\usepackage[T1]{fontenc}
\usepackage{amsmath}
\usepackage{hyperref}

\usepackage{bm}
\newcommand{\tr}[1]{\,\mathrm{tr}\left\lbrace  #1 \right\rbrace}
\newcommand{\ket}[1]{\vert#1\rangle}
\newcommand{\bra}[1]{\langle #1\vert}

\usepackage{tikz}
\usepackage{lipsum}

\newtheorem{theorem}{Theorem}
\newtheorem{lemma}[theorem]{Lemma}

\newenvironment{proof}[1][Proof]{\begin{trivlist}
\item[\hskip \labelsep {\bfseries #1}]}{\end{trivlist}}

\newcommand{\qed}{\nobreak \ifvmode \relax \else
      \ifdim\lastskip<1.5em \hskip-\lastskip
      \hskip1.5em plus0em minus0.5em \fi \nobreak
      \vrule height0.75em width0.5em depth0.25em\fi}

\begin{document}

\title{Hallmarking quantum states: unified framework for coherence and correlations}
\date{\today}
\author{Gian Luca Giorgi}
\affiliation{  Instituto de
F\'{\i}sica Interdisciplinar y Sistemas Complejos IFISC (UIB-CSIC), UIB Campus,
E-07122 Palma de Mallorca, Spain}
\email{gianluca@ifisc.uib-csic.es}
\author{Roberta Zambrini}
\affiliation{  Instituto de
F\'{\i}sica Interdisciplinar y Sistemas Complejos IFISC (UIB-CSIC), UIB Campus,
E-07122 Palma de Mallorca, Spain}

\maketitle

\begin{abstract}
 Quantum coherence and distributed correlations among subparties are often considered as separate, although operationally linked to each other, properties of a quantum state.  Here, we propose a measure able to quantify  the contributions derived by both the tensor structure of the multipartite Hilbert space and the presence of coherence inside each of the subparties. Our results hold for any number of partitions of the Hilbert space.  Within this unified framework,  global coherence of the state is identified as the ingredient responsible for the presence of distributed quantum correlations, while  local coherence also contributes to the quantumness of the state. A new quantifier, the ``hookup", is introduced within such a framework. We also provide a simple physical interpretation, in terms of coherence, of the  difference between total correlations and the sum of classical and quantum correlations obtained using relative-entropy--based quantifiers. 
\end{abstract}

\section{Introduction}
The superposition principle is one of the axioms and most distinctive features 
of quantum mechanics and it is responsible for the presence of coherence in 
quantum states  \cite{vn}.  When it comes to multipartite scenarios, the 
superposition principle is still the ingredient that makes quantum states 
intimately different from classical states and allows them to show inner 
correlations beyond any  classical probabilistic model 
\cite{horodecki,modirmp,adesso}.  

In the area of quantum information and quantum computation, there are many 
subfields where coherence and distributed correlations are seen as separate 
resources that enable quantum supremacy.   For instance, let us consider the  deterministic quantum computing with one qubit (DQC1) protocol, introduced by Knill and Laflamme   \cite{dqc1}. It was argued in Ref. 
\cite{datta} that the quantum computational advantage may be due to the 
production of quantum discord. More recently, it was shown that any discord 
quantifier is a witness for recoverable coherence, whose presence is necessary 
for the protocol to be successful \cite{matera}. As a second example, a  famous 
applications where the operational equivalence between local coherence and 
bipartite entanglement is quantum cryptography \cite{cryp}. Indeed, the original 
BB84 key distribution scheme, which makes use of ordinary
single-particle states in noncommuting bases \cite{bb84}, was proved to  have 
the same security bounds as  Ekert's scheme,  which relies on the use of 
maximally entangled states \cite{ekert}.

Despite the fact that coherence is one of the most distinctive traits of quantum 
mechanics, its characterization and quantification have only very recently 
become an intense field of investigation \cite{coherence}. In analogy to what has been 
done in different contexts, such as entanglement theory  
(\cite{horodecki,brandao,pleniovedral}) or quantum thermodynamics 
(\cite{gour,goold}), a resource theory for quantum coherence was proposed in 
Refs. \cite{aberg,levi,baumgratz,winter,misra,prx}. Coherence has also been 
linked to asymmetry \cite{asym1,asym2,asym3} and purity \cite{purity}.
The interplay between quantum correlations and coherence is mainly studied 
focusing on either the problem of interconversion between the two classes of 
resources \cite{streltsov2015,chimbatar,killoran,ma,qiao} or the distribution of 
coherence among subparties and its monogamy properties 
\cite{Hu,radha,yao,kumar,tan,kraft}.

Given the fact that quantum correlations and coherences are both useful resources in quantum information and computing, it appears very relevant to find a way to completely characterize
the overall computational power  of a  state exhibiting both aspects of quantumness through the introduction of a global quantifier (we will term it quantum hookup), which is the main scope of our paper. 
The way of pursuing this goal is represented by the introduction of a  unified framework  within which the full character of a quantum state belonging to a multipartite Hilbert space can be  determined by both local and collective properties. 

A unified framework was introduced in Ref. \cite{modi} trying to distinguish between total, classical, and quantum correlations (and also entanglement) as complementary parts of the same theoretical structure. 
In analogy to that approach, here we propose to merge the total power (the hookup) together with correlations and coherence in a consistent way. Within our framework, the concepts of classicality and quantumness need to be revisited.
 The founding observation to build the framework is that both coherence and correlations can be 
measured using the same kind of quantifier. 
  We will make use of the quantum relative 
entropy, but different metrics, such as the $l_1$-norm could be  be introduced 
as well  \cite{baumgratz}. As we will show, not only is our scheme useful to give a
comprehensive approach to quantumness through coherence, but it also allows one to explain the incongruent lack of closure  emerging in the framework of Ref. \cite{modi}, showing the indissoluble connexion between local and global quantum effects.

\section{Definitions}

Our starting point for quantifying correlations is the geometric scheme 
presented in Ref. \cite{modi}. All the distances between pairs of states $\rho$ 
and $\sigma$ are measured by the quantum relative entropy 
$S(\rho\vert\vert\sigma)=-\tr{\rho\log\sigma}-S(\rho)$, where 
$S(\rho)=-\tr{\rho\log\rho}$ is the von Neumann entropy of $\rho$. The relative 
entropy, in spite of its lack of symmetry,  is commonly used and accepted as a 
distance measure in different contexts, as it fulfils a series of important 
requirements \cite{relent}: it is a positive definite directed measure of the distance between two states;
it is contractive under completely positive and trace-preserving maps; its 
explicit calculation is feasible in various common scenarios; it avoids possible inconsistencies in the definition of discord \cite{piani,bruno}. 

Within this approach, the various kinds of correlations present in a quantum state are quantified by the distance between the state itself and the closest state without the desired property.
Thus, the total correlations of a multipartite state $\varrho\equiv\varrho_{A_1,A_2,\dots,A_n}$ are given by the distance from  the closest product state. 
Total correlations are quantified by the total mutual information of that state: 
\begin{equation}
{\cal T}(\varrho)=S(\varrho\vert\vert\pmb{ \pi}[\varrho])=S(\pmb{ \pi}[\varrho])-S(\varrho), 
\end{equation}
with $\pmb{ \pi}[\varrho]=\pmb{ \pi}_{A_1}\otimes\dots\otimes \pmb{ \pi}_{A_n}$ where $  \pmb{ \pi}_{A_i}$ is obtained from $\varrho$ calculating the partial trace over all the partitions with the exception of the $i$th. 
The quantum part of these correlations is measured by the (two-sided, relative entropy of) quantum discord  \cite{modi}
\begin{equation}
{\cal D}(\varrho)=S(\varrho\vert\vert \chi_\varrho)=S(\chi_\varrho)-S(\varrho),
\end{equation} 
where $ \chi_\varrho$ is the classically correlated state closest to $\varrho$ and where classically correlated states are separable in the $n$-partite Hilbert space:
\begin{equation}\label{basis}
\chi=\sum_{\vec{k}}p_{\vec{k}}\vert \vec{k}\rangle\langle \vec{k}\vert
\text{, with } \ket{\vec{k}}=\ket{k_1}\otimes\cdots\otimes \ket{k_n}.
\end{equation} 

The basis is given by product states, as the use of a nonlocal entangled basis would somehow hide the quantumness of the states within the basis itself.     
In turn, classical correlations of $\varrho$ (and equivalently of $\chi_\varrho$) are given by   
\begin{equation}
{\cal J}(\varrho)=S(\chi_\varrho\vert\vert  \pmb{ \pi}[\chi_\varrho])=S( \pmb{ \pi}[\chi_\varrho])-S(\chi_\varrho),
\end{equation}
 where $ \pmb{ \pi}[\chi_\varrho]$ is the product (uncorrelated) state closest to $\chi_\varrho$ \cite{modi}. 
Within this treatment, an  incongruity comes out, as in general the sum of classical plus quantum correlations exceeds the total correlations: ${\cal T}(\varrho)\le {\cal D}(\varrho)+{\cal J}(\varrho)$, with the equality holding only in some special cases. In fact, we have  \cite{modi}
\begin{equation}
L(\varrho)\equiv {\cal D}(\varrho)+{\cal J}(\varrho)-{\cal T}(\varrho)=S( \pmb{ \pi}[\varrho]\vert\vert \pmb{ \pi}[\chi_\varrho]).\label{elle}
\end{equation}
As we shall see later, we are able to give a physical interpretation for $L(\varrho)$  in terms of the quantum coherence of the local sub-parties. 

Here, it is worth remarking that, in this context, $ \chi_\varrho$ is commonly 
referred as ``the  classical state closest to  $\varrho$''
\cite{modi,vedral,luo}. Within a unified framework, $ \chi_\varrho$ more specifically identifies a 
\textit{classically correlated} state that can be coherent or not. Indeed there is only one 
special basis where its coherence vanishes and where it can be seen as a fully 
classical entity. Otherwise, a state exhibiting nonvanishing nondiagonal density-matrix 
elements can hardly be considered as a classical one but could be classically correlated. 

In analogy with correlations, also coherence can be quantified using the relative entropy \cite{baumgratz,singh}. 
The relative entropy of coherence of a state $\varrho$ with respect to a basis is defined as 
\begin{equation}
{\cal C}(\varrho)=\min_{\sigma \in {\cal I}} S(\varrho\vert\vert\sigma),\label{eqc}
\end{equation}
where $ {\cal I}$ is the set of totally incoherent (diagonal) states  in that basis. 
In general this definition is independent of possible partitions of the Hilbert space and, even for multipartite systems the basis could be local or nonlocal.
It turns out that  $
{\cal C}(\varrho)=S(\varrho\vert\vert{\bf \Delta}[\varrho])=S({\bf \Delta}[\varrho])-S(\varrho)$,
where  ${\bf \Delta}$ is the full decohering operation, leading to a state with all the non-diagonal elements of $\varrho$ set to zero \cite{baumgratz}. 
The dependence on the basis choice will not explicitly appear in the notation of ${\cal C},\varrho$ or other coherence related indicators. 
Henceforth, unless specifically indicated,  states will always be represented in product bases as in (\ref{basis}).
In particular, ${\bf \Delta}[\varrho]=\sum\ket{\vec{k}}\bra{\vec{k}}\varrho\ket{\vec{k}}\bra{\vec{k}}$. 

\section{Results}

\begin{figure}[t]
  \centering
   \includegraphics[scale=.5]{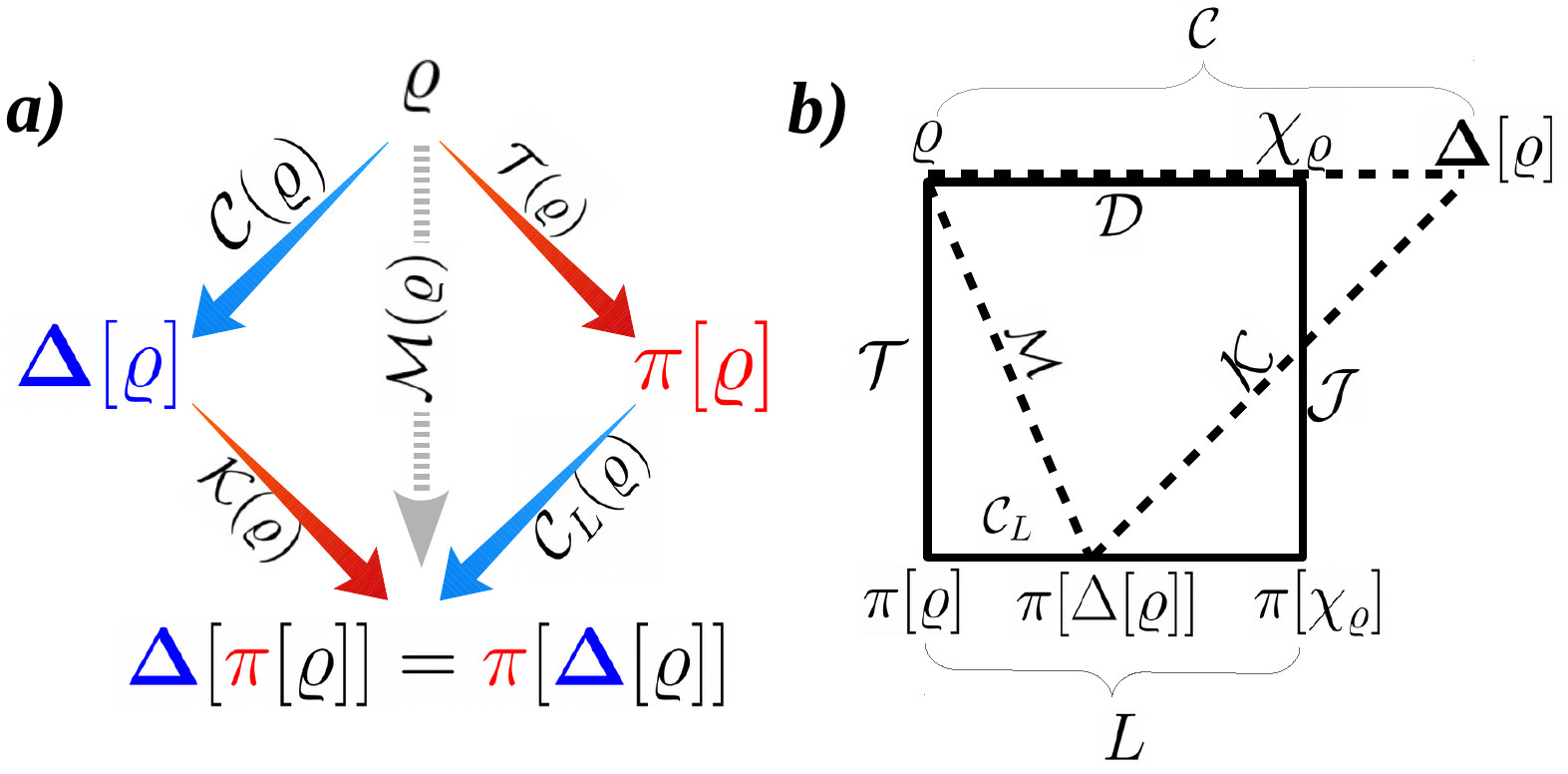}\\
  \caption{
(a) The unified framework of coherence and correlations: from any initial state $\varrho$, the closest useless state can be reached either by applying first the dephasing operation ${\bf \Delta}$ (blue) and then the product operation $\pmb{ \pi}$ (red) or the other way around. Depending on the chosen path, the total hookup ${\cal M}$ can be decomposed either into the sum of the total mutual information and the local coherence or into the sum of the total coherence and the irreducible classical information.  
  (b)  
  Comparison between the correlation scheme of Ref. \cite{modi} (solid lines) and the unified framework of coherence and correlations (dashed lines). While the discord is a lower bound for the total coherence (that is, the distance between $\varrho$ and $\chi_{\varrho}$ is always smaller than the one between $\varrho$ and ${\bf \Delta}[\varrho]$), ${\cal K}$ can be longer or shorter than ${\cal J}$ and ${\cal C}_L$ can be longer or shorter than $L$.  
    }\label{Fig1}
\end{figure}
As mentioned in the introduction, the distribution of coherence in multipartite settings has been subject of recent interest \cite{radha,yao,kumar,tan}. Building on these results, it is useful to introduce the concept of local coherence that is present in the state $\pmb{ \pi}(\varrho)$ to be distinguished from genuine multipartite effects \cite{radha,tan}. In the framework of relative entropy, the local coherence of a state $\varrho_{A_1,A_2,\dots,A_n}$ is the sum of the coherences of the reduced states  $ \pmb{ \pi}_{A_i}$: 
\begin{equation}
{\cal C}_L(\varrho)=\sum_{i=1}^n {\cal C}( \pmb{ \pi}_{A_i}),\label{lc}
\end{equation}
with ${\cal C}$ defined in Eq. (\ref{eqc}).
Using the additivity of the relative entropy, it can be shown that the following equality  holds: 
\begin{equation}
{\cal C}_L(\varrho)=S(\pmb{ \pi}[\varrho]\vert\vert {\bf \Delta}[ \pmb{ \pi}[\varrho]]).\label{lc2}
\end{equation}
In the same framework, a measure for the genuinely multipartite contribution to coherence can  be introduced by subtracting the contribution of local terms (see also Ref. \cite{maziero}): 
\begin{equation}
{\cal C}_M(\varrho)={\cal C}(\varrho)-{\cal C}_L(\varrho).
\end{equation}

 It can be shown that ${\cal C}_M(\varrho)$ is a nonnegative quantity,
and, for this purpose we anticipate the following [see also the illustration in Fig. \ref{Fig1}(a)]: 
 \begin{lemma}
The total dephasing operation ${\bf \Delta}[\cdot]$ commutes with $ \pmb{ \pi}[\cdot]$, that is, ${\bf \Delta}[ \pmb{ \pi}[\varrho]]= \pmb{ \pi}[{\bf \Delta}[\varrho]]$.
\end{lemma}

\begin{proof}
The proof can be given by  calculating explicitly the matrix elements of the two operators. Given a generic bipartite state (the proof is identical irrespective of the number of parties) $\varrho=\sum_{i,j,k,l}c_{i,j,k,l}\vert i,j\rangle\langle k,l\vert$, we have ${\bf \Delta}[\varrho]=\sum_{i,j}c_{i,j,i,j}\vert i,j\rangle\langle i,j\vert$ and $ \pmb{ \pi}[\varrho]=\sum_{i,j,k}c_{i,j,k,j}\vert i\rangle\langle k\vert \otimes \sum_{i,j,l}c_{i,j,i,l}\vert j\rangle\langle l\vert$. Thus, 
 for both operators we have ${\bf \Delta}[ \pmb{ \pi}[\varrho]]= \pmb{ \pi}[{\bf \Delta}[\varrho]]=\sum_{i,j}c_{i,j,i,j}\vert i\rangle\langle i\vert\otimes \sum_{i,j}c_{i,j,i,j}\vert j\rangle\langle j\vert$.\qed
\end{proof}

 The nonnegativity of ${\cal C}_M(\varrho)$, that is, the hierarchical relationship ${\cal C}(\varrho)\ge{\cal C}_L(\varrho)$, can be proved using the fact that ${\cal C}(\varrho)$ is the relative entropy between two states [$\varrho$ and $\sigma$ in (\ref{eqc})], while ${\cal C}_L(\varrho)$ is the relative entropy between two new states obtained from the previous ones by applying the quantum operation $ \pmb{ \pi}$ and by using  \textbf{Lemma 1} in Eq. (\ref{lc2}). We have
 \begin{equation}
 {\cal C}_M(\varrho)=S(\varrho\vert\vert{\bf \Delta}[\varrho])-S( \pmb{ \pi}[\varrho]\vert\vert \pmb{ \pi}[{\bf \Delta}[\varrho]]).\label{eqcma}
 \end{equation}
 As relative entropy is known not to increase under any completely positive, trace preserving (CPTP) quantum operation ($ \pmb{ \pi}$ in our case), ${\cal C}_{M}$ cannot be negative.  ${\cal C}_M(\varrho)$ also amounts to the difference between the total mutual information of $\varrho$ and the total mutual information of ${\bf \Delta}[\varrho]$. Indeed, ${\cal C}_M(\varrho)=S({\bf \Delta}[\varrho])-S(\varrho)-S( {\bf \Delta}[ \pmb{ \pi}[\varrho]])+S(\pmb{ \pi}[\varrho])$, which can be recombined as ${\cal C}_M(\varrho)=S(\varrho\vert\vert \pmb{ \pi}[\varrho])-S({\bf \Delta}[\varrho]\vert\vert{\bf \Delta}[ \pmb{ \pi}[\varrho]] )$ or
 \begin{equation}\label{eqcm}
 {\cal C}_M(\varrho)={\cal T}(\varrho)-{\cal T}({\bf \Delta}[\varrho]).
 \end{equation}

Interestingly, total correlations, measured by quantum mutual information, have a deep operational interpretation, 
being related to the work required to erase such correlations, that is, to convert a correlated state into a product one \cite{groisman}. 
This result was the generalization to quantum information of Landauer's theory of thermodynamics and information erasure \cite{landauer}. 
On the other hand, coherence has also been shown to be a resource in thermodynamical processes \cite{coherence}. 
Apart from being related to the work extraction problem \cite{korzekwa}, quantum coherence is at the root of fundamental issues, as, for example,  
irreversibility \cite{lostaglio}. 
Given these different roles played by quantum correlations and coherences, a legitimate question arises about %how much information is actually contained in
the global value of a generic quantum state displaying both. 
To answer this question it is necessary first to introduce a consistent unified framework within which  both the  correlations and the local quantum character of 
a state are taken into account and this is our aim here.

A unified framework for coherence and correlations can be built   identifying the class of states that are operationally useless, which is represented by the set of incoherent product states $ {\cal  \bar{I}}:\lbrace \rho \;\;{\rm s. t.}\; \; \rho=    \pmb{ \pi}[\rho]={\bf \Delta}[ \rho]  \rbrace$.
 Then, a meaningful measure for a state $\varrho $ is represented by its distance from its closest incoherent product state:
 \begin{equation}
{\cal M}(\varrho)=\min_{\sigma \in  {\cal  \bar{I}} }S(\varrho\vert\vert \sigma ).\label{eqm}
\end{equation}
We term ${\cal M}$ hookup, as it combines both quantum coherences and correlations, being related to the amount of noise necessary to erase both of them. 
\begin{theorem}
Given a state $\varrho$, its closest incoherent product state is obtained by applying the dephasing operation to  $\pmb{ \pi}[\varrho]$:
\begin{equation}
{\cal M}(\varrho)=S(\varrho\vert\vert {\bf \Delta}[ \pmb{ \pi}[\varrho] ]).
\end{equation}
\end{theorem}

 \begin{proof}
  The relative entropy between $\varrho$ and any incoherent state $\sigma$ can be written as $S(\varrho\vert\vert\sigma)=S({\bf \Delta}[\varrho])-S(\varrho)+S({\bf \Delta}[\varrho]\vert\vert\sigma)$. The product state closest to ${\bf \Delta}[\varrho]$ is $ \pmb{ \pi}[{\bf \Delta}[\varrho]]$ \cite{modi}, or, using \textbf{Lemma 1}, ${\bf \Delta}[ \pmb{ \pi}[\varrho]]$. Then, if $\sigma$ is also a product state, $S(\varrho\vert\vert\sigma)\ge  S({\bf \Delta}[\varrho])-S(\varrho)+S({\bf \Delta}[\varrho]\vert\vert{\bf \Delta}[\pmb{ \pi}[\varrho]])$.  Thus, $S(\varrho\vert\vert\sigma)\ge S({\bf \Delta}[ \pmb{ \pi}[\varrho]])-S(\varrho)={\cal M}(\varrho)$, where the equality holds  for $\sigma={\bf \Delta}[ \pmb{ \pi}[\varrho]]$.\qed
\end{proof}

Beyond the geometric definition of the hookup  ${\cal M}$ (\ref{eqm}), a clear physical interpretation
can be given observing that it can be decomposed as the sum of two terms, one of them associated to multipartite correlations 
[${\cal T} (\varrho)$] and the other one being, according to Eq. (\ref{lc}), the local coherence of the state ${\cal C}_L(\varrho)$ [see Fig. \ref{Fig1}(a)]
\begin{equation}
{\cal M}(\varrho)={\cal T} (\varrho)+{\cal C}_L(\varrho).\label{eqm1}
\end{equation}
In other words,  $\cal M$ is able to capture the resources of both correlations across the multipartite system and of the local coherences [Eq. (\ref{eqm1})]. 
Using Eq. (\ref{elle}) we can also decompose the hookup as
\begin{equation}\label{elle2}
{\cal M}(\varrho)={\cal D} (\varrho)+{\cal J} (\varrho)+{\cal C}_L(\varrho)-L(\varrho).
\end{equation}

Exploiting \textbf{Lemma 1}, a different decomposition of ${\cal M}$ can be given where the total coherence appears explicitly. In fact, we have [see Fig. \ref{Fig1}(a)]
\begin{equation}
{\cal M}(\varrho)={\cal C}(\varrho)+{\cal K} (\varrho),\label{eqm2}
\end{equation}
where  also the totally classical correlations ${\cal K} (\varrho)\equiv {\cal T} ({\bf \Delta}[\varrho])$ appear.  
The quantity ${\cal K} (\varrho)$, measuring total correlations of the diagonal ensemble, naturally emerges in our unified framework. 
Interestingly it has already been proved to play a relevant role in the context of many-body localization and quantum ergodicity \cite{gooldprb}. 
It quantifies the amount of information that survives to total dephasing and is given by the classical mutual information of $ \varrho $. 
It is a purely classical object, as it obtained by eliminating both quantum correlations and coherence. Thus, we will call ${\cal K}$ \textit{irreducible} classical information. 
It can also be written as ${\cal K}={\cal T}-{\cal C}_M$ [see (\ref{eqcm})].

Equation (\ref{eqm2}) redefines a different separation between the classical and the quantum part of resources with respect to (\ref{eqm1}). 
Indeed, the quantum content of the state is all contained in the total coherence ${\cal C}$, which is an upper bound for quantum discord. 
In fact, as already noticed in Ref. \cite{yao}, ${\cal D}(\varrho)$ is the minimum value of coherence calculated over all the possible local unitaries 
$U_{\rm loc}=U_1\otimes U_2\otimes\cdots \otimes U_n$ [see the upper distances in Fig. \ref{Fig1}(b)]. 
In other words, total coherence is minimized in the basis where $\chi_{\varrho}$ is completely incoherent ($\chi_{\varrho}={\overline{\bf \Delta}}[\varrho]$, 
where the bar reminds the special choice of basis adopted here).
This result can be understood observing that incoherent states are a sub-ensemble of the family of classical states.
As a consequence, $\overline{\bf \Delta}[\varrho]$ is the closest state to $\varrho$ which is completely classical, from the point of view of both coherence and correlations.  Then,
in the basis of the eigenstates of $\chi_{\varrho}$, $\chi_{\varrho}\equiv \overline{\bf \Delta}[\varrho]$, and ${\cal D}(\varrho)=\overline{\cal C}(\varrho)$.  
In such a  special basis, classical and irreducible classical correlations become equal as it will be shown below.

While $\cal J$ is expected to measure the amount of correlations that can be described within a classical probability model but does not 
take into account the freedom of changing the reference basis, which is a fully quantum property, $\cal K$ catches both aspects of classicality, that is, the absence of quantum correlations and the lack of coherence. 
The different definitions of $\cal J$ and $\cal K$ suggest a possible hierarchical relationship $\cal J \ge\cal K$. 
Actually, as explicitly shown in an example in Appendix \ref{Appendix}, this is not true. 
This is related to the mentioned lack of closure of Modi's scheme of correlations \cite{modi} [see Eq. (\ref{elle})], which may cause a bad estimation of classicality.

Interestingly, the excess term $L(\varrho)$  admits a clear physical interpretation within our framework. 
Indeed, let us consider the case  where $\varrho$ is written in the basis of the eigenstates of  $\chi_{\varrho}$. 
In that basis, according to what said before, $\chi_{\varrho}=\overline{\bf \Delta}[\varrho]$. Consequently, by applying \textbf{Lemma 1} and Eq. (\ref{elle}),
\begin{equation}\label{LcL}
 L(\varrho)=S(\overline{\bf \Delta}[ \pmb{ \pi}[\varrho]])-S(\pmb{ \pi}[\varrho])\equiv \overline{\cal C}_L(\varrho),
\end{equation}  
where again the bar reminds the special basis choice here.
This allows one to see that, as anticipated before, in the basis where the classical state $\chi_{\varrho}$
is incoherent, the hookup turns out to be the sum of discord and classical correlations  ${\cal M}(\varrho)={\cal D}(\varrho)+{\cal J}(\varrho)$.

The main consequence of Eq. (\ref{LcL}) is that a unified framework of correlations and coherences allows for 
a physical interpretation of $L$ as the local coherence of $\varrho$ calculated in the basis where $\chi_{\varrho}$ can be considered \textit{sensu stricto} classical. 
It is always possible to find a (classical) basis where ${\cal C}_L=0$ (such a basis is just the basis of the eigenvectors of $\pmb{ \pi}(\varrho)$). Then, a nonvanishing value of $L$ implies that the basis that minimizes the total coherence is a different one. It can be interpreted as a  basis mismatch measure and represents the lack of completeness of the correlation framework, as it tries to quantify quantumness by omitting the conceptually fundamental component of local coherence.

In order to understand the consequences of Eq. (\ref{eqm2}), let us discuss a simple example and consider the two-qubit state 
$\rho=\frac{1}{2}|\Phi^+\rangle\langle\Phi^+|+\frac{1}{4}|01\rangle\langle 01|+\frac{1}{4}|10\rangle\langle 10|$, where $|\Phi^+\rangle=(|00\rangle+|11\rangle)/\sqrt{2}$
in the computational basis. The closest  incoherent state is the identity operator $\mathbbm{I}/4$, which is obviously completely uncorrelated. This implies that for such a state 
$\mathcal{M}=\mathcal{C}=\mathcal{C}_M=0.5$, while $\mathcal{K}=0$. Thus, all the usable resources are contained in the coherences of the Bell state and have a purely quantum nature, 
while there is no classical contribution to them. The state has finite classical correlations  $\mathcal{J}(\rho)\simeq 0.19$, obtained by applying the decohering operator in the rotated
$x$ basis, which gives $\chi_\rho=\mathbbm{I}/4+(|00\rangle\langle 11|+|10\rangle\langle 01|+h.c.)/8$, but all these correlations are due to the initial presence of multipartite coherence. In fact, also the 
coherence of $\chi_\rho$ is genuinely multipartite: $\mathcal{J}(\chi_\rho)=\mathcal{C}(\chi_\rho)=\mathcal{C}_M(\chi_\rho)$.
We point out that this example shows that the coherence of a classically correlated state can be (even completely) multipartite. 
Furthermore, discord (here $\mathcal{D}(\rho)\simeq 0.31$) can be present even for vanishing irreducible classical correlations, at difference from the classical ones \cite{mdms}.

Previously, we have commented on the lack of hierarchy between $\cal K$ and $\cal J$. As a complementary aspect, the same lack of hierarchy takes place between $L$ and ${\cal C}_L$. The case ${\cal C}_L\ge L$ can be found considering for instance any pure state, as $L=0$
(in fact, the optimal dephasing basis to find $\chi$ is always the basis of the eigenstates of the state itself).  A case where this ordering relationship is violated can be found by considering the state $\upsilon=\frac{8}{27}|000  \rangle\langle 000 |+\frac{12}{27}|W  \rangle\langle W |+\frac{6}{27}|\bar{W}  \rangle\langle \bar{W}  |+\frac{1}{27}| 111 \rangle\langle  111|$ \cite{wei}, where $|\bar{W}  \rangle=\frac{1}{\sqrt{3}}(|011\rangle+|110\rangle+|101\rangle)$. In fact, we have $L(\upsilon)=0.24$ \cite{modi}. It can immediately checked that the local coherence vanishes in the computational basis and it can also be shown that ${\cal C}_L(\upsilon)\le L(\upsilon)$ for any local basis change.

Finally, let us mention that, apart from the relationship between discord and coherence, which are both indicators of overall quantumness, it is also possible to establish a similar one between the purely multipartite coherence ${\cal C}_M(\varrho_{A_1,\dots,A_n})$ and an indicator of genuine quantum correlations, the so called global discord ${\cal G}(\varrho_{A_1,\dots,A_n})$ \cite{global}, defined as ${\cal G}(\varrho_{A_1,\dots,A_n})=\min\limits_{\{\Phi^i\}}{\cal G}_{\{\Phi^i\}}(\varrho_{A_1,\dots,A_n})$, with ${\cal G}_{\{\Phi^i\}}(\varrho_{A_1,\dots,A_n}) = S(\varrho_{A_1,\dots,A_n}||\Phi^i(\varrho_{A_1,\dots,A_n}))-\sum_k S(\varrho_{A_k}||\Phi^i_{A_k}(\varrho_{A_k}))$, 
where $\Phi^i =\otimes_{j=1}^n\Phi_{A_j}^{i_j}$ is the dephasing operator in a classical basis and where
 $\Phi^i_{A_k}(\varrho_{A_k})=\sum_i |i_k\rangle\langle i_k|\varrho_{A_k} |i_k\rangle\langle i_k|$. It is immediate to note that the global discord  represents  a lower bound for   the multipartite coherence:
\begin{equation}
{\cal G}(\varrho)=\min_{U_{\rm loc}} {\cal C}_M(\varrho),
\end{equation} 
 where the minimum is taken over the set local unitaries.

%%%%%%%%%%%%%%

%%%%%%%%%%%%%%%%%

\section{Conclusions}

To summarize, when it comes to characterize a quantum state, the (total) mutual 
information is not necessarily an adequate indicator, for it fails to take into 
account the coherence properties of the state itself. This is the reason why the 
puzzling term $L(\varrho)$ comes out in the relative entropy framework. Such a 
term can be taken as a witness of the fact that the local coherence and the 
global one are minimized in different bases. 
This seemingly side observation actually reveals how much a unified framework is 
needed to build a consistent theory.

In the approach proposed here, quantum coherence and multipartite correlations 
cannot be thought as distinct labels, as they both contribute to hallmark the state, 
 providing a full description of its quantumness. 
%have defined a new measure which gives a complete characterization of the 
We have introduced the hookup ${\cal M}$ and shown that it amounts to the sum of 
coherence and irreducible classical information resources or equivalently to the sum of local coherence and total correlations. 
Such a comprehensive hallmark 
has different conceivable applications, as it fully determines the power of a 
quantum state. As previously  mentioned, the hookup ${\cal M}$ can be used to 
measure the amount of work necessary to erase such correlations and has obvious 
thermodynamic implications that can be further explored, for instance in the 
field of ergotropy.
 The interplay between quantum local and distributed contributions could also be 
employed in computing tasks, as the  algorithmic performances are studied by 
analyzing either the single-qubit power or the presence of correlations as 
entanglement or discord.  Another field where correlations and coherence are 
separately essential resources is quantum metrology, where our approach can be 
used to find the optimal quantum advantage and tighter bounds.
A possible extension of our framework could concern the use of entanglement to quantify the multipartite quantumness instead of discord. In this case, the main obstacle is represented by the difficulty to define the closest unentangled state in terms of elementary operation, as the partial trace or the total dephasing.
Finally,  a further line of research that a unified framework can open concerns 
the possibility of converting different types of resources among them 
\cite{ma,qiao}.

\begin{acknowledgments}
Valuable discussions with Fernando Galve, Francesco Plastina, Alex Streltsov, and Manuel Gessner are kindly acknowledged.  
This work was supported by the EU through the H2020 Project QuProCS (Grant Agreement 641277), by MINECO/AEI/FEDER through projects  EPheQuCS FIS2016-78010-P,  by the Mar\'{i}a de Maeztu Program for Units of Excellence in R\text{$\&$}D (MDM-2017-0711),  and by  the Conselleria d’Innovaci\'{o}, Recerca i Turisme del Govern de les Illes Balears.
\end{acknowledgments}

\appendix

\section{Comparison between $\mathcal{K}$ and $\mathcal{J}$}\label{Appendix}

 As an instructive example of the lack of hierarchy between $\mathcal{K}$ and $\mathcal{J}$, let us consider the maximally discordant mixed state $\varrho_{MD}=\epsilon\vert \Phi^{+}  \rangle\langle  \Phi^{+}   \vert+ (1-\epsilon)+\vert 10  \rangle\langle  10 \vert$ (the choice is suggested by the fact that this state is known to have low classical correlations compared to discord) \cite{mdms}, together with the whole family of states obtained by applying local unitaries. Such a family is given by 
$\tilde\varrho_{MD}=U \varrho_{MD}U^{\dag}$, where the most general form for $U$ is $U=U_1\otimes U_2$ with
\begin{equation}
U_j=\left(
\begin{array}{cc}
\cos\theta_j& e^{i \phi_j}\sin\theta_j\\
-e^{-i \phi_j}\sin\theta_j&\cos\theta_j
\end{array}
\right).
\end{equation}
The state is symmetric under the exchange of the two qubits up to the spin-flip operation. Such a symmetry is reflected into the extremal values assumed by the angles $\theta_j$ in optimal measurements \cite{chiribella},
as calculating the von Neumann entropy of $\tilde\varrho_{MD}$ amounts to calculating the entropy of $\varrho_{MD}$ in a rotated basis obtained by applying the unitary $U^{\dag}$ to the set of computational states. Thus, the set of optimal values is restricted to $\theta_1=\theta_2\equiv\theta$ and $\phi_1=-\phi_2$. It can be further shown that the result is independent on the phases, provided that   $\phi_1=-\phi_2$, and then they can be fixed to zero.

  Depending on the value of $\epsilon $,  the closest classically correlated state $ \chi_{\varrho_{MD}
 }$ is obtained by dephasing $\varrho_{MD}$ either in the computational basis for $\epsilon< \epsilon^{\prime}$ or in a rotated basis for $\epsilon> \epsilon^{\prime} $, where the threshold is given by $\epsilon^{\prime}=2/3$. Above a second threshold, for  $\epsilon>\epsilon^{\prime\prime}\simeq 0.76$, the optimal basis is the $x$-basis. On the other hand,  ${\cal K}$ reaches its maximum in the $x$-basis irrespective of $\epsilon$.  This means that, for $\epsilon< \epsilon^{\prime\prime}$, ${\cal K}(\tilde\varrho_{MD})$ can be either bigger or smaller than ${\cal J}(\varrho_{MD})$ depending on the unitary chosen, while ${\cal K}(\tilde\varrho_{MD})\le{\cal J}(\varrho_{MD})$ for $\epsilon > \epsilon^{\prime\prime}$. As for the hookup ${\cal M}$, it always reaches its maximum in the $x$-basis and its minimum in the computational basis.

\begin{figure}[t]
  \centering
  % Requires \usepackage{graphicx}
  \includegraphics[width=0.47 \textwidth]{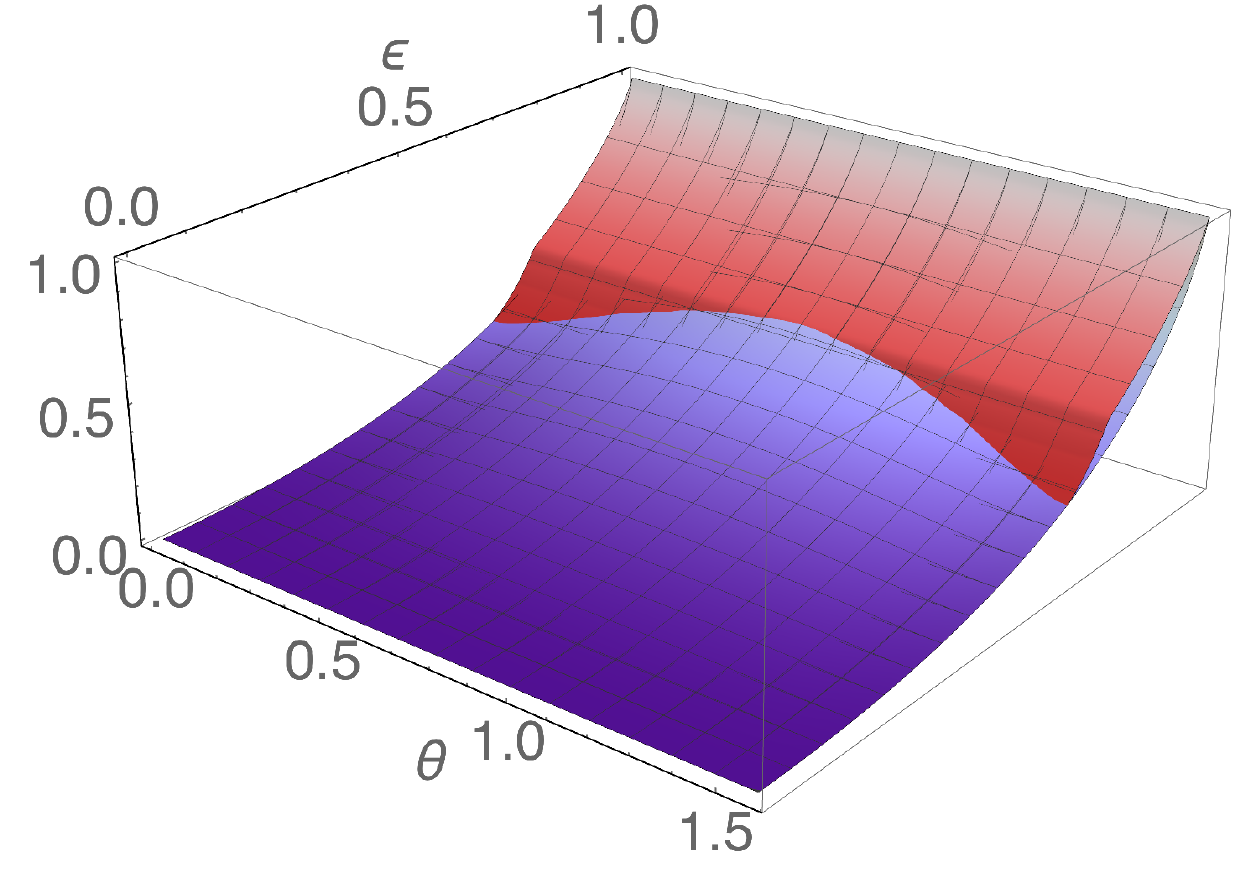}\\
  \caption{  
  Comparison between ${\cal J}$ (red) and ${\cal K}$  (blue) for the family $\tilde\varrho_{MD}$ as a function of $\theta$ and $\epsilon$.  ${\cal J}$ is an upper bound for ${\cal K}(\theta)$ only in the region $\epsilon^{\prime\prime}<\epsilon<1$ (see text).
    }\label{Fig2}
\end{figure}
In Fig. \ref{Fig2}, we compare ${\cal J}$ and ${\cal K}$. 
The analytical values of $\epsilon^{\prime}$ and  $\epsilon^{\prime\prime}$ can be obtained by solving, respectively, the equations 
\begin{eqnarray}
&&\lim_{\theta\to 0}\frac{\partial^2 S(\tilde\varrho_{MD})}{\partial\epsilon^2}=0,\\
&&\lim_{\theta\to \pi/4}\frac{\partial^2 S(\tilde\varrho_{MD})}{\partial\epsilon^2}=0.
\end{eqnarray}

\bibliographystyle{plain}

\begin{thebibliography}{9}


\bibitem{vn} J. von Neumann, \href{https://doi.org/10.2307/j.ctt1wq8zhp} {\textit{ Mathematical Foundations of Quantum Mechanics}, (Springer, Berlin, 1932). }



\bibitem{horodecki} R. Horodecki, P. Horodecki, M. Horodecki, and K. Horodecki, \href{https://doi.org/10.1103/RevModPhys.81.865}{Rev. Mod. Phys. \textbf{81}, 865 (2009).}

\bibitem{modirmp} K. Modi, A. Brodutch, H. Cable, T. Paterek, and V. Vedral, \href{https://doi.org/10.1103/RevModPhys.84.1655}{Rev. Mod. Phys. \textbf{84}, 1655 (2012).}

\bibitem{adesso} G. Adesso, T. R. Bromley, and M. Cianciaruso, \href{https://doi.org/10.1088/1751-8113/49/47/473001} {J. Phys. A: Math. Theor. 
\textbf{49}, 473001 (2016).}

\bibitem{dqc1} E. Knill and R. Laflamme,  \href{https://doi.org/10.1103/PhysRevLett.81.5672} {Phys. Rev. Lett. \textbf{81}, 5672 (1998).}

\bibitem{datta} A. Datta, A. Shaji, and C. M. Caves,  \href{https://doi.org/10.1103/PhysRevLett.100.050502} {Phys. Rev. Lett. \textbf{100}, 050502 (2008).}

\bibitem{matera} J. M. Matera, D. Egloff, N. Killoran, and M. B. Plenio, \href{https://doi.org/10.1088/2058-9565/1/1/01LT01}  {Quantum Sci. Technol. \textbf{1}, 01LT01 (2016).}

\bibitem{cryp} N. Gisin, G. Ribordy, W. Tittel, and H. Zbinden,  \href{https://doi.org/10.1103/RevModPhys.74.145} {Rev. Mod. Phys. \textbf{74}, 145 (2002).}


\bibitem{bb84} C. H. Bennett and G. Brassard, in \textit{Proceedings of the IEEE
International Conference on Computers, Systems, and Signal
Processing, Bangalore, India} (IEEE, New York, 1984), p. 175; republished in \href{https://doi.org/10.1016/j.tcs.2014.05.025} {Theor. Comput. Sci. \textbf{560}, 7 (2014).}




\bibitem{ekert} A. K. Ekert,  \href{https://doi.org/10.1103/PhysRevLett.67.661} {Phys. Rev. Lett. \textbf{67}, 661 (1991).}


\bibitem{coherence}  A. Streltsov, G. Adesso, and M. B. Plenio,  \href{https://doi.org/10.1103/RevModPhys.89.041003} {Rev. Mod. Phys. \textbf{89}, 041003 (2017). }

\bibitem{brandao} F. G. S. L. Brand\~{a}o and M. B. Plenio,  \href{https://doi.org/10.1038/nphys1100
} {Nature Phys. \textbf{4}, 873 (2008).}




\bibitem{pleniovedral} V. Vedral and M. B. Plenio, \href{https://doi.org/10.1103/PhysRevA.57.1619}  {Phys. Rev. A \textbf{57}, 1619 (1998).}


\bibitem{gour} G. Gour, M. P. M\"{u}ller, V. Narasimhachar, R. W. Spekkens, and N. Y. Halpern,  \href{https://doi.org/10.1016/j.physrep.2015.04.003} {Phys. Rep. \textbf{583}, 1 (2015).}

\bibitem{goold} J. Goold, M. Huber, A. Riera, L. del Rio, and P. Skrzypczyk, \href{https://doi.org/10.1088/1751-8113/49/14/143001}  {J. Phys. A \textbf{49}, 143001 (2016).}



\bibitem{aberg} J. Aberg, arXiv:quant-ph/0612146.

\bibitem{levi} F. Levi  and F. Mintert, \href{https://doi.org/10.1088/1367-2630/16/3/033007}  {New J. Phys. \textbf{16}, 033007 (2014).}

\bibitem{baumgratz}  T. Baumgratz, M. Cramer, and M. B. Plenio,  \href{https://doi.org/10.1103/PhysRevLett.113.140401} {Phys. Rev. Lett. \textbf{113}, 140401 (2014).}

\bibitem{winter}  A. Winter and D. Yang, \href{https://doi.org/10.1103/PhysRevLett.116.120404}  {Phys. Rev. Lett. \textbf{116}, 120404 (2016).}

\bibitem{misra}  A. Misra, U. Singh, S. Bhattacharya, and A. K. Pati, \href{https://doi.org/10.1103/PhysRevA.93.052335} 
{Phys. Rev. A \textbf{93}, 052335 (2016).}

\bibitem{prx} A. Streltsov, S. Rana, M. N. Bera, and M. Lewenstein, \href{https://doi.org/10.1103/PhysRevX.7.011024} 
{Phys. Rev. X \textbf{7}, 011024  (2017).}

\bibitem{asym1}  I.  Marvian and R. W.  Spekkens,  \href{https://doi.org/10.1038/ncomms4821} {Nat. Commun.  \textbf{5}, 3821   (2014).}

\bibitem{asym2} 
Y. Yao, G. H. Dong, X. Xiao, and C. P. Sun,  \href{https://doi.org/10.1038/srep32010} 
	{Sci. Rep. \textbf{6}, 32010 (2016).}

 \bibitem{asym3} I.  Marvian and R. W.  Spekkens,  \href{https://doi.org/10.1103/PhysRevA.94.052324} {Phys. Rev. A. \textbf{94}, 052324 (2016). }

\bibitem{purity} A. Streltsov, H. Kampermann, S. W\"olk, M. Gessner, and D. Bru\ss,	\href{https://doi.org/10.1088/1367-2630/aac484} {New J. Phys. \textbf{20}, 053058 (2018).}




\bibitem{streltsov2015} A. Streltsov, U. Singh, H. S. Dhar, M. N. Bera, and G. Adesso, 
 \href{https://doi.org/10.1103/PhysRevLett.115.020403} {Phys. Rev. Lett. \textbf{115}, 020403 (2015).}

\bibitem{chimbatar} E. Chitambar, A. Streltsov, S. Rana, M. N. Bera, G. Adesso, and M. Lewenstein
 \href{https://doi.org/10.1103/PhysRevLett.116.070402} {Phys. Rev. Lett. \textbf{116}, 070402 (2016).}



\bibitem{killoran} N. Killoran, F. E. S. Steinhoff, and M. B. Plenio, \href{https://doi.org/10.1103/PhysRevLett.116.080402} {Phys. Rev. Lett. \textbf{116}, 080402 (2016).}


\bibitem{ma} J. Ma, B. Yadin, D. Girolami, V. Vedral, and M. Gu, \href{https://doi.org/10.1103/PhysRevLett.116.160407} 
{Phys. Rev. Lett. \textbf{116}, 160407  (2016).}

\bibitem{qiao}L.-F. Qiao et al., arXiv:1710.04447.


\bibitem{Hu} M.-L. Hu, X. Hu, J.-C. Wang, Y. Peng, Y.-R. Zhang, and  H. Fan, \href{https://doi.org/10.1016/j.physrep.2018.07.004} {Phys. Rep. \textbf{762–764}, 1-100  (2018)}.



\bibitem{radha} C. Radhakrishnan, M. Parthasarathy, S. Jambulingam, and T. Byrnes, \href{https://doi.org/10.1103/PhysRevLett.116.150504} 
{Phys. Rev. Lett. \textbf{116}, 150504 (2016).}

\bibitem{yao} Y. Yao, X. Xiao, L. Ge, and C. P. Sun,  \href{https://doi.org/10.1103/PhysRevA.92.022112} {Phys. Rev. A \textbf{92}, 022112 (2015).}

\bibitem{kumar} A. Kumar, \href{https://doi.org/10.1016/j.physleta.2017.01.046}  {Phys. Lett. A \textbf{381}, 991 (2017).}


\bibitem{tan} K. C. Tan, H. Kwon, C.-Y. Park, and H. Jeong, \href{https://doi.org/10.1103/PhysRevA.94.022329} 
{Phys. Rev. A \textbf{94}, 022329  (2016).}

\bibitem{kraft} T. Kraft and M. Piani, \href{https://doi.org/10.1088/1751-8121/aab8ad}  {J. Phys. A: Math. Theor. \textbf{51}, 41401 (2018).}


\bibitem{modi} K. Modi, T. Paterek, W. Son, V. Vedral, and M. Williamson, \href{https://doi.org/10.1103/PhysRevLett.104.080501} 
{Phys. Rev. Lett. \textbf{104}, 080501 (2010).}

\bibitem{relent} V. Vedral,  \href{https://doi.org/10.1103/RevModPhys.74.197} {Rev. Mod. Phys. \textbf{74}, 197 (2002).}

\bibitem{piani} M. Piani, \href{https://doi.org/10.1103/PhysRevA.86.034101} {Phys. Rev. A \textbf{86}, 034101 (2012).}

\bibitem{bruno} B. Bellomo, G. L. Giorgi, F. Galve, R. Lo Franco, G. Compagno, and R. Zambrini, \href{https://doi.org/10.1103/PhysRevA.85.032104}
{Phys. Rev. A \textbf{85}, 032104 (2012).}


\bibitem{vedral} L.  Henderson and V.  Vedral,  \href{https://doi.org/10.1088/0305-4470/34/35/315}  {J. Phys. A  {\bf 34}, 6899 (2001).}


\bibitem{luo} N. Li and S. Luo, \href{https://doi.org/10.1103/PhysRevA.78.024303} {Phys. Rev. A \textbf{78}, 024303 (2008).}



\bibitem{singh} U. Singh, M. N. Bera, A. Misra, and A. K. Pati, arXiv:1506.08186.


\bibitem{maziero} M. B. Pozzobom and J. Maziero, \href{https://doi.org/10.1016/j.aop.2016.12.031}  {Ann. Phys. \textbf{ 377}, 243 (2017).}


\bibitem{groisman} B. Groisman, S. Popescu, and A. Winter, \href{https://doi.org/10.1103/PhysRevA.72.032317}   {Phys. Rev. A  {\bf 72}, 032317 (2005).}

\bibitem{landauer} R. Landauer,  \href{https://doi.org/10.1147/rd.53.0183} { IBM J. Res. Dev. \textbf{5}, 183 (1961). }

\bibitem{korzekwa} K. Korzekwa, M. Lostaglio, J. Oppenheim,  and D. Jennings,  \href{https://doi.org/10.1088/1367-2630/18/2/023045} {New J. Phys. \textbf{18}, 023045 (2016).}


\bibitem{lostaglio} M. Lostaglio, D. Jennings, and T. Rudolph, \href{https://doi.org/10.1038/ncomms7383}  {Nat. Commun. \textbf{6}, 6383 (2015).}


\bibitem{gooldprb} J. Goold, C. Gogolin, S. R. Clark, J. Eisert, A. Scardicchio, and A. Silva, \href{https://doi.org/10.1103/PhysRevB.92.180202}  {Phys. Rev. B \textbf{92}, 180202 (2015).}


\bibitem{mdms} F. Galve, G. L. Giorgi, and R. Zambrini,  \href{https://doi.org/10.1103/PhysRevA.83.01210} {Phys. Rev. A \textbf{83}, 012102 (2011).}

 
\bibitem{wei} T. C. Wei, M. Ericsson, P. M. Goldbard, and W. J. Munro,  {Quantum Inf. Comput. \textbf{4}, 252 (2004).}


\bibitem{global} C. C. Rulli and M. S. Sarandy, \href{https://doi.org/10.1103/PhysRevA.84.042109} {Phys. Rev. A \textbf{84}, 042109 (2011).}




\bibitem{chiribella} G. Chiribella and G. M. D'Ariano, \href{ https://doi.org/10.1063/1.2349481} {J. Math. Phys \textbf{47}, 092107 (2006).}
\end{thebibliography}

\end{document}